\documentclass[review,11pt]{elsarticle}
\usepackage{color}
\usepackage{comment}
\usepackage{stackengine}

\usepackage{graphics}
\usepackage{caption}
\usepackage{subcaption}
\usepackage{graphicx}

\usepackage{amsmath,amsfonts, amssymb}
\usepackage{amsthm}

\usepackage{stfloats}

\usepackage{algorithm}
\usepackage{algpseudocode}

\usepackage{array}

\def\R{\mathbb{R}} 
\def\C{\mathbb{C}} 
\def\x{\mathbf{x}} 
\def\vv{\mathbf{v}} 
\def\u{\mathbf{u}}
\def\y{\mathbf{y}}
\def\w{\mathbf{w}}

\def\Ww{$W_{\|\cdot\|_2}$}

\newtheorem{definition}{Definition}
\newtheorem{theorem}{Theorem}
\newtheorem{Proposition}{Proposition}

\newtheorem{Lemma}{Lemma}

\journal{Signal Processing}

\begin{document}

\begin{frontmatter}

\title{Classifying Sums of Exponentially Damped Sinusoids Using an Associated  Numerical Range}
%
%


\author[fiuba]{Raymundo~Albert\corref{cor1}}
\ead{araymundo@fi.uba.ar}

\author[csc,fiuba]{Cecilia~Galarza\corref{cor1}} 
\ead{cgalarza@csc.conicet.gov.ar}

\cortext[cor1]{Corresponding author}
\address[csc]{Centro de simulaci\'on computacional, CONICET, Godoy Cruz 2390, C1425FQD Buenos Aires, Argentina}
\address[fiuba]{Facultad de Ingenier\'ia, Universidad de Buenos Aires, Av. Paseo Colón 850, C1063ACV Buenos Aires, Argentina}

\begin{abstract}
The matrix pencil method (MPM) is a well-known technique for estimating the parameters of exponentially damped sinusoids in noise by solving a generalized eigenvalue problem. However, in several cases, this is an ill-conditioned problem whose solution is highly biased under small perturbations. When the estimation is performed to classify the observed signal into two categories, the estimation errors induce several misclassifications. In this work we propose a novel signal classification criteria by exploiting the relationship between the generalized eigenvalue problem posed in the MPM and the numerical range of a pair of rectangular matrices. In particular, the classification test is formulated as a set inclusion problem, and no spectrum estimation is required. The technique is applied to a problem of electromagnetic scattering to classify dielectric materials using the scattering signal observed when a target is illuminated by an ultra-wideband signal. The performance of the classification scheme is assessed in terms of error rate and it is compared to another classification technique, the generalized likelihood rate test (GLRT). 
\end{abstract}


\end{frontmatter}

\section{Introduction}

Problems dealing with a mixture of closely spaced sinusoidal signals in a noisy environment are regularly observed in engineering. This is the case of the scattering collected from a dielectric target illuminated by an Ultra-WideBand (UWB) signal, where the particular mixture is a feature that characterizes the target. Identification of the complex frequencies embedded in the scattered signal becomes relevant for target identification and classification. Motivated by this problem, we have designed a novel strategy for classifying mixtures of sinusoidal signals. In particular, this procedure exploits the underlying structure among the complex frequencies of the observed signal without actually identifying their values. 

\subsection{Related Work}

Traditional approaches to this problem include the matrix pencil method (MPM)\cite{Hua1990}. This is a spectrum estimation technique that constructs a generalized eigenvalue problem from a partition of the Hankel matrix of the observed signal. In this setup, the matrix pencil is non-square and the eigenvalue computation can become unstable or the solutions may fail to exist. This problem is particularly acute when the number of sinusoids in the original mixture is very large and the complex frequencies are close to each other in the complex plane.  These facts make very difficult to accurately identify the parameters of the exponentially damped sinusoidal signals.

An alternative was studied in \cite{Mooney1998,Chen2014}, where a generalized likelihood test was posed to classify perfectly conducting objects using scattered signals. 
 Nevertheless, GLRT is sensitive to unmodelled dynamics, which could be the case when a large number of sinusoids form the observed mixture. 

On the other hand, research on eigenvalue computation for a matrix leads to the analysis of its numerical range, which is a set in the complex plane that contains the spectrum of the matrix \cite{Horn1991,Bonsall1971}. Analysis of the numerical range of matrices has been an active field of research in the past. Some of those results have set the ground for the classification criteria that we present in this paper. In particular, we return to the definition for the numerical range of rectangular matrices introduced in \cite{Chorianopoulos2009}. The authors of this paper defined the numerical range of a pair of rectangular matrices that resulted in a compact and convex set of the complex plane. 

\subsection{Contributions}


Making a hinge on the MPM, in particular on the associated generalized eigenvalue problem, we formulate a signal classification criteria using the numerical range for the matrices of the matrix pencil. In particular, the classification is formulated as a set inclusion problem, hence avoiding the need to compute eigenvalues of probably large ill-conditioned matrices. 

We start our paper with a review of the MPM in section \ref{Sect:MPM}. In particular, we discuss how the MPM poses a generalized eigenvalue problem that may be complicated to solve accurately. In section \ref{Sect:NR}, we reproduce the definition for the numerical range of a pair of rectangular matrices and we show the properties that are relevant to our development. Along these lines, we show that the rectangular numerical range of a pair of matrices contains the corresponding generalized eigenvalues. Making ground on this observation, in section \ref{Sect:Main_Results} we propose a new classification technique that is solved as a set inclusion problem. To check the validity of our approach, we analyze in section \ref{sec:NumRange} the behavior of the rectangular numerical range related to noisy signals for different values of the signal to noise ratio (SNR). Finally, in section \ref{Sect:Num_results} we show numerical results when classifying sums of exponentially damped sinusoids. Also, we apply the new procedure to scattering signals from different dielectric materials and compare its performance to an already known classification technique, the Generalized Likelihood Ratio Test (GLRT).

\section{Matrix Pencil Method}
\label{Sect:MPM}

	The Matrix Pencil Method (MPM) has been used as a high-resolution spectrum estimation technique on different applications for several years now \cite{Li1997,Lu1998}. In this section, we review briefly the main aspects of this technique. 
	
	Suppose that we want to analyze a signal that is composed by a sum of damped complex exponentials perturbed by noise. We assume that different looks of a single experiment are acquired. For example, in the case of the electromagnetic scattering, an array of antennas spatially distributed may be collecting the scattered signals from different observation angles. On a general framework, we consider the signal vector $\y_t$ with $K$ components
	
	\begin{equation}
		\y_t = \sum_{i=1}^{M} \mathbf{c}_i z_i^t + \mathbf{w}_t, \qquad t = 0,1, \ldots
		\label{Eq:MPM1}
	\end{equation}

\noindent where $K$ is the number of looks, $M$ is the number of exponentials, $z_i \in \C$ is a complex resonant frequency, and $\mathbf{c}_i\in \C^K$ is the vector of residues associated to $z_i$.  

Now, for two integers $s$ and $n$, $s>n$, consider $s+n-1$ samples of $\y_t$. Let $m=sK$ and define $H_{\y}$ as the following $(m \times (n+1))$-block Hankel matrix

	\begin{equation}
		H_{\y} = \begin{bmatrix}
				\y_0       & \y_1     & \cdots & \y_{n}\\[0.3em]
				\y_1       & \y_2     & \cdots & \y_{n+1}\\[0.3em]
				\vdots     & \vdots   & \ddots & \vdots\\[0.3em]
				\y_{s-1} & \y_{s} & \cdots & \y_{s+n-1}
			\end{bmatrix} \in \C^{m\times (n+1)}.
		\label{Eq:MPM3}
	\end{equation}

	The integer $n$ is known as the pencil parameter and it is a resource to mitigate the effect of $\w_t$ on $\y_t$\footnote{It has been shown that when using the MPM for estimating $z_i$, its variance is minimum for values of $n$ between $s/2$ and $2s$ \cite{Hua1990}.}, provided that $M\leq n \leq s-M$. 
	
	A well-known result states that in the noiseless case, the rank of $H_{\y}$ is $M$ \cite{Gantmacher1960}. To show that, recall that the entries of $H_{\y}$ satisfy the following recurrence relation of order $M$ 
	\begin{equation}
		\y_t = \sum_{k = 1}^{M}a_k\y_{t-k},	\quad t = M, M+1,\ldots ,
		\label{Eq:MPM9}
	\end{equation}
and $M$ is the least integer number satisfying \eqref{Eq:MPM9} for all $t$. Hence, every column of $H$ is a linear combination of the first $M$ columns, and $H$ has indeed rank equal to $M$. 
	
	Define now two $(m\times n)$-block matrices, $A_{\y}$ and $B_{\y}$, where the first one is obtained by deleting the first column of $H_{\y}$, and the second one is obtained from $H_{\y}$ by deleting its last column. A matrix pencil is then defined as 
	\begin{equation}
		A_{\y} - \lambda B_{\y}.
		\label{Eq:MPM10}
	\end{equation}

	When the signal $\y_t$ is noiseless, the rank of the pencil \eqref{Eq:MPM10} is $M$ as long as $\lambda \neq z_i$ , $i = 1,\ldots,M$. When  $\lambda = z_i$, the rank of the pencil is reduced by one. Furthermore, $z_i, i=1, \ldots, M$ can be computed as the generalized eigenvalues of $(A_{\y}, B_{\y})$. Alternatively, when $B_{\y}$ is full rank, $z_i$ solves the eigenvalue problem 
		\begin{equation}
		B_{\y}^\dagger A_{\y} - \lambda I_n,
		\label{Eq:MPM10a}
	\end{equation}
\noindent where $B_{\y}^\dagger$ is the pseudoinverse and  $I_n$ is the identity matrix of order $n$. This is the principle for using the MPM for estimating the values of $z_i$. 

However, when considering noisy signals, the pencil \eqref{Eq:MPM10} has usually full rank for any complex scalar $\lambda$. To overcome this problem, several procedures have been proposed in the past. All of them obtain an approximation for the original Hankel matrix before computing the pencil. A first simple approach was followed in \cite{Hua1991}, where the singular value decomposition of the original Hankel matrix was computed and only the directions corresponding to the dominant singular values were kept. The authors in  \cite{Li1997,Lu1998} noticed that this procedure did not guarantee a Hankel structure for the resulting approximation of $H_{\y}$. Following \cite{Cadzow1988}, they proposed an iterative procedure that kept the Hankel structure in the approximated matrix. Alternatively, in \cite{Park1999,Dicle2013}, the problem was tackled by solving a total least-squares problem. In both procedures, knowledge of $M$ was key, and it had to be estimated when not known a priori. 

Nevertheless, even when the value of $M$ is known, the computation of generalized eigenvalues may be a delicate problem. For a better understanding of this issue, and following \cite{Boutry2005}, we define the function $g : \C\to\R_{\geq 0}$
		
		\begin{equation}
			g(\lambda) = \sigma_{min}\{A_{\y}-\lambda B_{\y}\}
		\end{equation}
		where $\sigma_{min}$ is the smallest singular value. For a pencil $A_{\y}-\lambda B_{\y}$ that has a rank-reducing solution at $\lambda = z$, $g(z)=0$. Therefore, the generalized eigenvalues of $(A_{\y},B_{\y})$ are the local minima of this function. 
	
		Consider now a signal as in \eqref{Eq:MPM1} with $K=1$ and $M=4$, under two scenarios: one without noise ($\w_t=0$), and one with noise ($\w_t\neq 0$). Using $n=M$, we build $g^{noiseless}(\lambda)$ from the noiseless signal and $g^{noisy}(\lambda)$ from the noisy one. Fig. \ref{Fig:SingularValueMin_Noiseless} shows the level curves of function $g^{noiseless}(\lambda)$. Observe that in this case, its minima coincide with the location of $z_i, i=1, \ldots 4$.  
		
		\begin{figure}[ht]
			\centering
			\includegraphics[width = .6\textwidth]{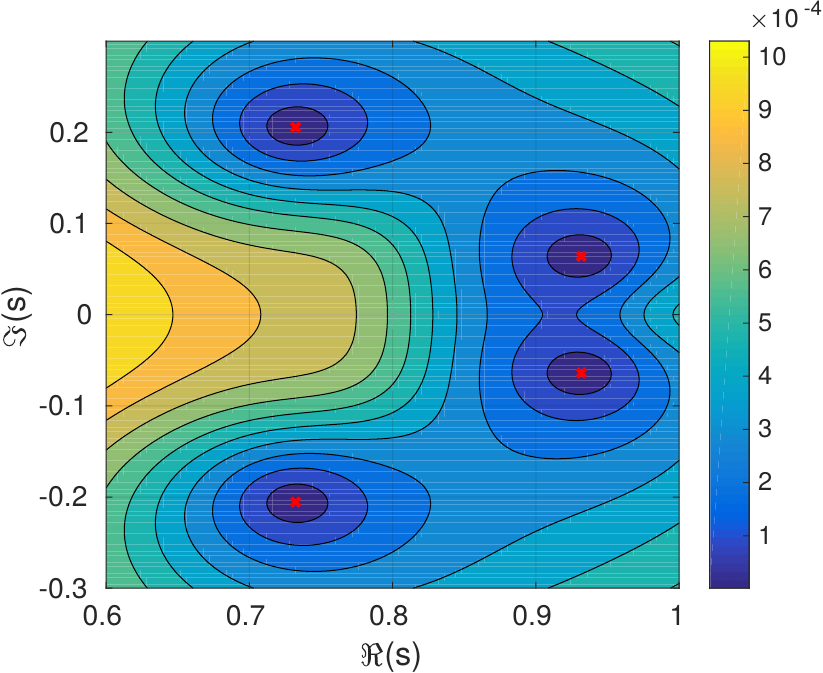}
			\caption{Function $g^{noiseless}(\lambda)$, and $z_1, \ldots z_4$ ('$\times$').}
			\label{Fig:SingularValueMin_Noiseless}
		\end{figure} 
		
	Now in Fig. \ref{Fig:SingularValueMin_Noise}, we have plotted the level curves for $g^{noisy}(\lambda)$ for a signal to noise ratio 30 dB. Notice that in this case, the local minima move away from the actual values of $z_1$ through $z_4$. It is clear from this very simple example that there are major difficulties that we may face when identifying the oscillation modes from noisy signals. These difficulties are even more serious when the value of $M$ is unknown and the order of the model needs to be identified beforehand \cite{Stoica2004,Gavish2014}.
			
		\begin{figure}[ht]
			\centering 
			\includegraphics[width = .6\textwidth]{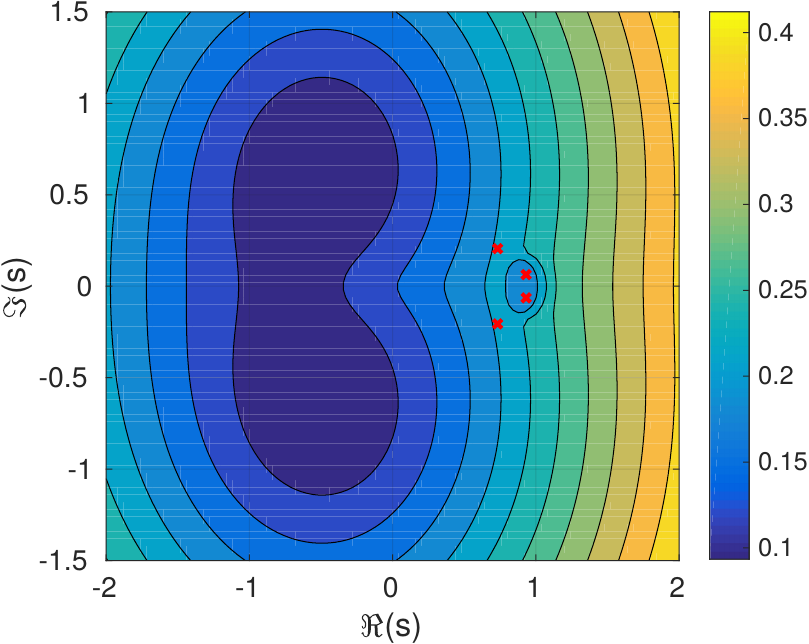}
			\caption{Function $g^{noisy}(\lambda)$, and $z_1, \ldots z_4$ ('$\times$').}
			\label{Fig:SingularValueMin_Noise}
		\end{figure}

\section{The numerical range of a matrix pencil}\label{Sect:NR}

As shown previously, solving \eqref{Eq:MPM1} may be difficult when dealing with perturbed matrices. Nevertheless, it is interesting to define a region on the complex plane where the solutions would be. A candidate for that is the numerical range associated to $(A_{\y}$, $B_{\y})$. In this section, we present background material necessary to present the main result of the paper. It is not our intention to provide a thorough analysis of the properties and problems related to the numerical range of a matrix. For an in-depth tutorial on this matter, we refer the reader to \cite{Horn1991}. 

Let $A$ be a complex $n\times n$-matrix. Its numerical range or field of values is defined as
	\begin{equation}
		W(A) := \big\{\x^H A\x : \x\in\C^n, \x^H\x = 1\big\} \subset \C
		\label{Eq:NumRange}.
	\end{equation}
	

\textbf{Properties of $W(A)$}
\begin{enumerate}
\item All the eigenvalues of $A$ lie in $W(A)$. 
\item $W(A)$ is a closed convex set.
\item It has the form \cite{Bonsall1971}
	\begin{equation}
		\begin{aligned} 
			W(A) &= \big\{\theta\in\C:\|A-\lambda I_n\|_2 \geq |\theta - \lambda|,\forall\lambda\in\C\big\}.
		\end{aligned}
		\label{Eq:NumRange2}
	\end{equation}
\end{enumerate}
	
\noindent where $\|\cdot\|_2$ is the spectral norm. The set $W(A)$ appears naturally when recursively computing the eigenvalues of $A$. For instance, the Arnoldi recurrence, which is the universally used algorithm for this matter\footnote{A variant of the Arnoldi algorithm has been implemented in ARPACK.}, approximates the sought eigenvalues with elements of $W(A)$. 
	
The concept of eigenvalues is generalized when dealing with two rectangular matrices. Let $A,B\in\C^{m\times n}$. A scalar $\beta$ is a generalized eigenvalue of $(A,B)$ if there exists $\x\in \C^n$ such that
	\begin{equation}
		A\x = \beta B\x, \qquad \x \neq \mathbf{0}.
		\label{Eq:Pencil}
	\end{equation}
Following \eqref{Eq:NumRange2}, we introduce the definition for the numerical range of the pencil $(A;B)$ as
	\begin{align}
			W_{\|\cdot\|_2}(A;B) & = \big\{\theta\in\C:\|A-\lambda B\|\geq |\theta - \lambda|,\forall \lambda\in\C\big\} \nonumber \\
			& =  \bigcap_{\lambda\in\C}\mathcal{D}(\lambda,\|A-\lambda B\|).
		\label{Eq:NumRange3}
	\end{align}
This definition was first proposed in \cite{Chorianopoulos2009}. We state, without proofs, some properties of $W_{\|\cdot\|_2}(A;B)$. 

\begin{Proposition}
$W_{\|\cdot\|_2}(A;B)$ is non-empty if and only if $\|B\|\geq 1$. 
\label{Prop:WnonEmpty}
\end{Proposition}	
Notice that from  \eqref{Eq:NumRange3}, we have that $\forall \theta \in W_{\|\cdot\|_2}(A;B)$, $|\theta|\leq \|A\|_2$. Moreover, for $\lambda$ such that $|\theta| < |\lambda|$, the following is true
\[ |\lambda|-|\theta| \leq \|A-\lambda B\|_2\leq \|A\|_2+|\lambda|\|B\|_2.
\]
Hence $|\lambda|(1-\|B\|_2) \leq 2\|A\|_2$. If $\|B\|_2<1$, this is satisfied by $\lambda$ on a bounded set only, and not on $\C$. Therefore, $\theta \notin W_{\|\cdot\|_2}(A;B)$. Then if $\|B\|_2<1$, for any $\theta \in W_{\|\cdot\|_2}(A;B)$ it is possible to find $\lambda$ that does not satisfy \eqref{Eq:NumRange3}. We conclude that $W_{\|\cdot\|_2}(A;B)=\emptyset$. It can be proved that $\|B\|_2>1$ is also a necessary condition for non-emptiness.
		
\begin{Proposition}
Any generalized eigenvalue $\beta$ of $(A,B)$ with associated eigenvector $\x\in \C^n$ such that $\|B\x\|\geq 1$, lies in $W_{\|\cdot\|_2}(A;B)$.
\label{Th:NumRange1}
\end{Proposition}			
			
\noindent Here, the condition $\|B\x\|\geq 1$ is only sufficient but not necessary. Experimentally, it has been verified that the generalized eigenvalues lied in $W_{\|\cdot\|_2}(A;B)$ although $\|B\x\| < 1$. 

\begin{Proposition}
Let $\alpha\in \C$ such that $|\alpha|\geq 1$. Then for any $A,B \in \C^{m\times n}$,
\[
W_{\|.\|_2}(A,B)\subset W_{\|.\|_2}(\alpha A,\alpha B).
\]
\end{Proposition}

\begin{Proposition}
\label{Prop:frob_norm}
When the Frobenious norm is used in \eqref{Eq:NumRange3}, and $\|B\|_F\geq 1$, 
	\begin{equation}
		W_{\|\cdot\|_F}(A;B) = \mathcal{D}\bigg(\frac{\langle A, B\rangle}{\|B\|_F^2},\bigg\|A - \frac{\langle A,B\rangle}{\|B\|_F^2}B\bigg\|_F\frac{\sqrt{\|B\|_F^2-1}}{\|B\|_F}\bigg).
		\label{Eq:NumRange4}
	\end{equation}
\end{Proposition}

\noindent In this case, the numerical range is characterized by a closed disc centered at $\frac{\langle A, B\rangle}{\|B\|^2}$, which is a superset that contains $W_{\|\cdot\|_2}(A,B)$, because of the equivalence between $\|.\|_2$ and $\|.\|_F$. 

	\begin{Lemma}\label{TH:2}
		For any $A,B\in\C^{m\times n}$ $(m\geq n)$, $\|B\|_2 = 1$, it holds that \Ww$(A;B)\subseteq W(B^\dagger A)$.
		\label{Th:NumRange2}
	\end{Lemma}
				
	\begin{proof}
		When $m>n$ and using the definition of the Moore-Penrose pseudoinverse $B^\dagger$, we have that
		\[ \begin{aligned}
				\text{\Ww}(A;B) & = \big\{\theta\in\C : \|(B^\dagger A-\lambda I_n)B \|_2\geq|\theta-\lambda|,\forall\lambda\C\big\}\\
				& \subseteq \{\theta\in\C : \|B^\dagger A - \lambda I_n\|_2\geq|\theta-\lambda|,\forall\lambda\in\C\} \\
				& = W(B^\dagger A),
			\end{aligned}
		\]
		where the last equality follows from \eqref{Eq:NumRange2}.
	\end{proof}
Recall that Eq. \eqref{Eq:MPM10a} formulates the generalized eigenvalue problem as a regular one using $B^\dagger A$. As a consequence, if $\beta$ in \eqref{Eq:Pencil} exists, it lies in $W(B^\dagger A)$. In the particular case of Lemma \ref{Th:NumRange2}, we have that \Ww$(A;B)\subseteq W(B^\dagger A)$.

\section{Main Result}\label{Sect:Main_Results}
		
\subsection{A Classification Problem}\label{Sect:CP}
\label{subsec:classproblem}

Our quest is for a criterion to classify pairs of rectangular matrices $(A,B)$ according to their generalized eigenvalues. A direct approach to this problem would compute the generalized eigenvalues. However, this would require lengthy computations when considering large size matrices, and the overall approach might fail to obtain a valid solution when using noisy experimental data. To construct a more robust strategy, we propose to use the set \Ww$(A;B)$ to characterize the behavior of  $(A;B)$, rather than the actual values of the eigenvalues. Resorting to a minor abuse of notation, we will say that a matrix class $\Theta$ is defined by a finite set of complex numbers that would also be called $\Theta$.

\begin{definition}\label{Def:1}
Let  $\Theta = \{\theta_{1},\theta_{2},\ldots,\theta_{p}\}\subset \C$. Consider two matrices $A,B \in C^{m\times n}$, with $\|B\|_2\geq 1$. We say that  $(A,B)$ is  of class $\Theta$ if and only if $\theta_{i} \in W_{\|\cdot\|_2}(A;B)$ for all $i = 1,\ldots,p$.  

The set $\{\theta_{1},\theta_{2},\ldots,\theta_{p}\}$ is the candidate set, $\Theta$ is the candidate class, and $(A;B)$ is the observed pair. 
\end{definition}

The classification procedure requires verifying whether $\theta_{i}$ is a member of $W_{\|\cdot\|_2}(A;B)$ or not. Next theorem shows how to answer this question by solving a minimization problem.

	\begin{theorem}\label{TH:1}
		Given $\theta\in\C$, define the function $f_\theta:\C\to\R$ 		
		\[
			f_\theta(\lambda) = \|A-\lambda  B\|_2-|\theta-\lambda|.
		\]
		Then, $\theta\in W_{\|\cdot\|_2}(A;B)$ if and only if 		
		\[	
			\inf_{\lambda\in\C} f_\theta(\lambda) \geq 0
		\]
	\end{theorem}
				
	\begin{proof}
	Using \eqref{Eq:NumRange3}, we have that the following statements are equivalent:
		\[
			\begin{aligned}
				&\theta\in W_{\|\cdot\|_2}(A;B) \Leftrightarrow \\
				& \|A-\lambda  B\|_2\geq |\theta-\lambda|, \forall \lambda\in\C \Leftrightarrow\\
				& \|A-\lambda  B\|_2-|\theta-\lambda| \geq 0, \forall \lambda\in\C \Leftrightarrow\\
				& f_\theta(\lambda) \geq 0, \forall \lambda\in\C.
			\end{aligned}		
		\]
The result follows from here. 
	\end{proof}

\subsection{Classification of exponentially damped sinusoids}	
Suppose now that the measured outcome of an experiment is modeled as in \eqref{Eq:MPM1}. In particular, consider that the set of complex resonant frequencies, $\{z_i, i=1, \ldots, M\}$ , is a feature that characterizes the experimental observation. As such, we associate a class of observed signals to a set of complex frequencies. For instance, when observing a scattering phenomenon, the class of signals represents the material from which the illuminated target is made of. Then, given a finite sample collection of $\y_t$, the question at hand is whether this observation lies within the class defined a priori. 

To proceed along this path, we assume that we know a priori all or a subset of the complex resonant frequencies present in the signals of interest. That will be, we know $Z=\{z_i, i=1, \ldots p\}$, where $p \leq M$. Consider now that we collect $s+n-1$ samples of the observed signal $\y_t$ and we use them to build the Hankel matrix $H_{\y}$ and the associated matrices $A_{\y}\in C^{m\times n}$ and $B_{\y}\in C^{m\times n}$ as in section \ref{Sect:MPM}. A straightforward extension of Definition \ref{Def:1} leads to the following definition

\begin{definition}\label{Def:2}
Let $Z = \{z_{1},z_{2},\ldots,z_{p}\}\subset \C$. Consider the matrices $A_{\y},B_{\y} $ built from $\y_t, t=0, \ldots s+n-1$. We say that  $\y_t$ is  of class $Z$ if and only if $z_{i} \in W_{\|\cdot\|_2}(A_{\y};B_{\y})$ for all $i = 1,\ldots,p$.  

The set $Z$ is the candidate set that characterizes the candidate class, $\y_t$ is the observed signal, and $(A_{\y}, B_{\y})$ is the observed pair. 
\end{definition}

Definition \ref{Def:2} allows us to classify the observed signal without explicitly solving the generalized eigenvalue problem formulated in the MPM. We call this procedure the Classification by Rectangular Numerical Range (CRNR). Algorithm \ref{Algo} shows an implementation of this classification criteria.

\begin{algorithm}
	\begin{algorithmic}[1]
		\State {Given $\y_t$, $Z = \{z_{1},z_{2},\ldots,z_{p}\}$, and a scalar $D>1$}
		\State{Obtain a reduced-rank approximation for $H_{\y}$}
		\label{step:approxH}
		\State {Compute $(A_{\y},B_{\y})$ from $H_{\y}$ .}
		\If {$\|B\|_2 <1$}
			\State{$A_{\y}=\frac{D}{\|B_\y\|_2}A_{\y}$} \label{step:scaleA}
		\State {$B_{\y}=\frac{D}{\|B_\y\|_2}B_{\y}$ } \label{step:scaleB}		
		\EndIf 
		\For {$k = 1,\ldots,p$}
			\If {$r(A_{\y},B_{\y}) < |z_k-c(A_{\y},B_{\y})|$} \label{step:inclusion_normF}
				\State{$\y_t$ is not of class $Z$.}
				\Return
			\Else
				\State { $\delta_k = \inf_{\lambda\in\C} \|A_{\y}-\lambda  B_{\y}\|_2 - |z_{k}-\lambda|$}
				\If { $\delta_k < 0 $}
					\State {$\y_t$ is not of class $Z$.}
					\Return
				\EndIf
			\EndIf
		\EndFor
		\State {$\y_t$ is of class $Z$.}
	\end{algorithmic}
	\caption{}
	\label{Algo}
\end{algorithm}

When considering a single candidate class, we commit a classification error if $Z \subset  W_{\|\cdot\|_2}(A_{\y};B_{\y})$ given that $\y_t$ is not from class $Z$. When $c$ candidate classes are considered, we deal with $c$ sets $Z_k$, $k = 1,\ldots, c$. The construction of each one  depends on the signals being classified. On the other hand, when we observe $\y_t$ of class $Z_j$, and consider the candidate class $Z_k$, with $k\neq j$, we commit a classification error when $Z_k \subset  W_{\|\cdot\|_2}(A_{\y};B_{\y})$. This observarion will be used to characterize the performance of the classification criteria.



In the general case, the condition $\|B_{\y}\|_2\geq 1$ for having a non-empty rectangular numerical range for the pair $(A_{\y},B_{\y})$ cannot be guaranteed. However, an adequate scaling of $\y_t$ may solve this problem. On the other hand, since any complex frequency $z_i$ of $\y_t$ is also a complex frequency of $\alpha \y_t$ for any complex scalar $\alpha$, scaling is a suitable resource for applying the results of subsection \ref{subsec:classproblem} to the classification of a sum of complex exponentials. The appropriate scaling factor depends on the set $Z$. One may think that $\alpha$ is selected during a setup process and kept fixed afterward. 

A sensible step when dealing with noisy data is to obtain a reduced-rank approximation for $H_\y$. In this paper, we follow the ideas from \cite{Cadzow1988}. For completeness, we have included the procedure in \ref{appendix1}. We will analyze the relevance of this step in the following section. 

As mentioned in Proposition \ref{Prop:frob_norm}, $W_{\|.\|_2}$
is a subset of $W_{\|.\|_F}$, which is a disk as in \eqref{Eq:NumRange4}. Let $c(A_{\y},B_{\y})$ and $r(A_{\y},B_{\y})$ be the center and radius of that disk. Exploiting the structure of $A_\y$ and $B_\y$ we obtain the following closed-form expressions:
\begin{gather}
	c(A_{\y},B_{\y})=
	 \frac{\sum_{i=0}^{m+n-2}\eta_i\y_{i+1}^*\y_i}{\sum_{i = 0}^{m+n-2}\eta_i|\y_i|^2}
	\label{Eq:centre} \\
	r(A_{\y},B_{\y})=\bigg[
	\sum_{i = 0}^{m+n-2}\eta_i|\y_{i+1}|^2- \frac{|\sum_{i=0}^{m+n-2}\eta_i\y_{i+1}^*\y_i|^2}{\sum_{i = 0}^{m+n-2}\eta_i|\y_i|^2}\bigg]^{\frac{1}{2}}  \notag \\ 
 	\times 
 	 \sqrt{
 \frac{\sum_{i = 0}^{m+n-2}\eta_i|\y_i|^2-1}{\sum_{i = 0}^{m+n-2}\eta_i|\y_i|^2}},
 	 \label{Eq:radii}
\end{gather}
where $\eta_i$ is the number of elements in the $i-$th anti-diagonal of $A_\y$ (or $B_\y$). Any $z$ such that 
$|z-c(A_{\y},B_{\y})|>r(A_{\y},B_{\y})$, does not lie within $W_{\|.\|_F}(A_{\y},B_{\y})$. Hence $z \notin W_{\|.\|_2}(A_{\y},B_{\y})$. In Algorithm \ref{Algo} we include this simple check to quickly reject the candidate frequencies outside $W_{\|.\|_F}(A_{\y},B_{\y})$. Following Th. \ref{TH:1}, for any other candidate frequency, we compute
\[
\delta = \inf_{\lambda \in \C}\|A_{\y}-\lambda B_{\y}\|_2 -|z-\lambda|.
\]
If $\delta<0$, then $z\notin W_{\|.\|_2}(A_{\y},B_{\y})$.

\section{Numerical range for noisy signals}
\label{sec:NumRange}


The observed signal $\y_t$ has a noise component that will alter the boundaries of $W_{\|.\|_2}(A_\y, B_\y)$. In general, large noise terms tend to expand the boundaries of the numerical range. To mitigate this problem, we perform a reduced-rank approximation of $H_\y$ as in step \ref{step:approxH} in Algorithm \ref{Algo}. In this section, we analyze a particular example to show the relevance of this procedure by studing the behavior of $W_{\|.\|_F}(A_\y, B_\y)$ for different noise levels. 

We build a signal as in \eqref{Eq:MPM1}, using $K=1$, $M=10$, and $z_{i}\in Z_1$ defined as
\begin{equation}
\begin{aligned} 
	Z_1 =  \big\{& 0.4474 \pm 0.5822j; 0.4447 \pm 0.5782j; 0.4236 \pm 0.5874j;\\
	& 0.4166 \pm 0.5959j; 0.3871 \pm 0.5858j\big\}.\end{aligned} 			\label{Eq:Z1} 
\end{equation}
The noise term $\w_t$ is a white Gaussian process with known variance. The signal to noise ratio (SNR) is defined as usual as the ratio between the signal power to the noise power. For each SNR, we perform a Montecarlo experiment with 10000 realizations. For each record of $\{\y_t, t=0, \ldots, s+n-1\}$, we compute \eqref{Eq:centre} and \eqref{Eq:radii}.

Figures \ref{Fig:Centre} and \ref{Fig:radius}  show the ensemble averages of $c(A_{\y},B_{\y})$ and $r(A_{\y},B_{\y})$, for each SNR, before and after performing the reduced-rank approximation. 

\begin{figure}[ht]
	\centering
	\includegraphics[width = .6\textwidth]{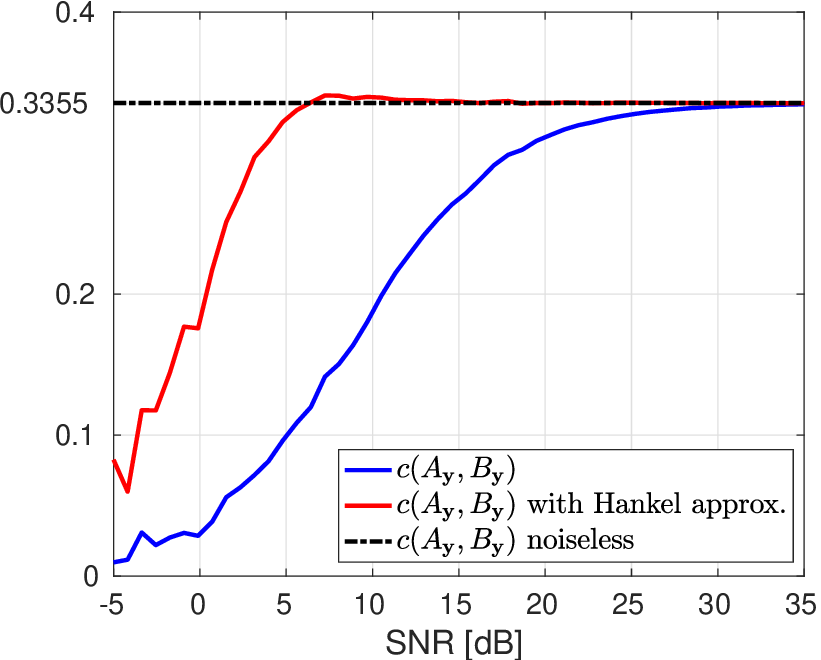}
	\caption{ $c(A_{\y},B_{\y})$ vs. SNR with and without reduced-rank Hankel approximation.}
	\label{Fig:Centre}
\end{figure}

We observe from Fig. \ref{Fig:Centre} and \ref{Fig:radius}  that  $W_{\|.\|_F}(A_\y, B_\y)$ gets enlarged as the SNR deteriorates. For instance, the expected value of $r(A_{\y},B_{\y})$ grows as the noise power increases. However, the effect is dramatically reduced when we perform the reduced-rank approximation of $H_\y$. Also, the center $c(A_{\y},B_{\y})$ remains close to its noiseless value when the Hankel approximation is performed. In general, the rank reduction procedure contributes to mitigate the effect of noise on the boundaries of the numerical range. Since $W_{\|.\|_2}(A_\y, B_\y)$ is a subset of $W_{\|.\|_F}(A_\y, B_\y)$, similar conclusions hold for the numerical range defined with the 2-norm. This is an important feature, since the classification process is defined by the boundaries of $W_{\|.\|_2}(A_\y, B_\y)$ and its ability to enclose the complex frequencies from one class and to leave aside the complex frequencies from other classes. 

Although we are showing these results on a particular example, the same conclusions were observed elsewhere. An important note to make is that this was the case when the approximated $H_\y$ kept the Hankel structure. Simple truncation of the singular value decomposition of $H_\y$ did not perform adequately. 
		
	\begin{figure}[ht]
		\centering
		\includegraphics[width = .6\textwidth]{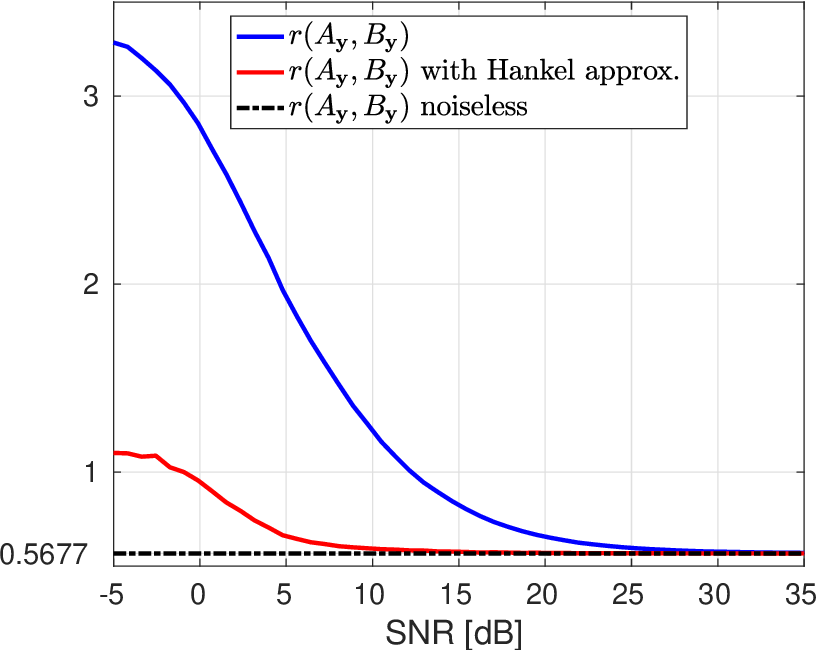}
		\caption{$r(A_{\y},B_{\y})$ vs. SNR with and without reduced-rank Hankel approximation.}
		\label{Fig:radius}
	\end{figure}

\section{Numerical results}
\label{Sect:Num_results}	

\subsection{Generalized likelihood test}
Using the statistical description of $\y_t$, we can formulate a hypothesis testing problem and solve it by using a likelihood ratio test. We will use the GLRT procedure to assess the performance of the classification methodology proposed above. In this section, we succinctly describe the application of the GLRT in the context described so far. 

Let the observed signal $\y_t$ be from class $Z_1$ as before, and consider a second class of signals, $Z_2$ as the candidate class. Each signal in each class has its own vector of residues in \eqref{Eq:MPM1}. As such,  $\mathbf{c}_{i,1}$ is related to signals of class $Z_1$, and $\mathbf{c}_{i,2}$ to signals of class $Z_2$. In general, $\mathbf{c}_{i,1}$ and $\mathbf{c}_{i,2}$ are  unknown vectors. To robustify the classification, a generalized likelihood test (GLRT) is formulated as in \cite{Mooney1998}. Given $Z_1 = \{z_{i,1}, i=1, \ldots p\}$ and $Z_2 = \{z_{i,2}, i=1, \ldots p\}$, and 
the observed signal $\y_0, \ldots, \y_l$, $l=s+n-1$, we define the hypothesis $\mathbb{H}_1$ and $\mathbb{H}_2$ as follows
\[
\mathbb{H}_i : \left[\begin{array}{c}
\y_0^t \\ \vdots \\\y_l^t
\end{array} \right] = F_i \left[\begin{array}{c}
\mathbf{c}_{1,i}^t \\ \vdots \\ \mathbf{c}_{p,i}^t 
\end{array} \right], \text{ where }
F_i = \left[\begin{array}{ccc}
z_{1,i}^0 & \cdots & z_{p,i}^0  \\ 
 & \ddots & \\
z_{1,i}^l & \cdots & z_{p,i}^l  \\ 
\end{array} \right] ,
\]
is a matrix formed from the complex frequencies associated with $Z_1$ and $Z_2$ respectively. Notice that $z_{k,i}$, $k=1, \ldots, p$, $i=1, 2$, are known frequencies. Then, the GLRT is posed as 
\[
\frac{\max_{\mathbf{c}_{i,1}}p(\y_0, \ldots, \y_l|H_1)}{\max_{\mathbf{c}_{i,2}}p(\y_0, \ldots, \y_l|H_2)} \begin{aligned} & H_1 \\[-5pt]
& \gtreqless\\[-5pt]
& H_2
\end{aligned} \quad 1 .
\]

\subsection{Known model order}	

In this section, we test the classification procedure when the model order is known a priori. For that, we 
build the observed signal $\y_t$ from class $Z_1$ as before for different SNRs. For each SNR, the algorithm selects an appropriate scaling factor as stated in Algorithm \ref{Algo}.  
We consider the candidate class $Z_2$
		\begin{equation}
			\begin{aligned}
			Z_2 =  \big\{ & 0.0429 \pm 0.0825j; -0.4130 \pm 0.1176j;-0.3118 \pm 0.2127j;\\ & -0.1951 \pm 0.3642j;-0.3385 \pm 0.1249j\big\}.
			\end{aligned} 
			\label{Eq:Z2}
		\end{equation}

Fig. \ref{Fig:A1_B1} shows a particular realization of $W_{\|.\|_2}(A_\y,B_\y)$ for different SNRs. We have also incorporated on this figure the sets $Z_1$ and $Z_2$. Notice that $Z_1$ is always within the boundaries of $W_{\|.\|_2}(A_{\y},B_{\y})$. In this case, the set $Z_2$ is close to the boundaries of \Ww$(A_{\y};B_{\y})$ but it is not included on it. This may not be the case on all the realizations. Erroneous classifications are expected on those realizations.

		\begin{figure}[ht]
			\centering
			\begin{subfigure}[t]{0.22\textwidth}
				\centering
				\includegraphics[width =\textwidth]{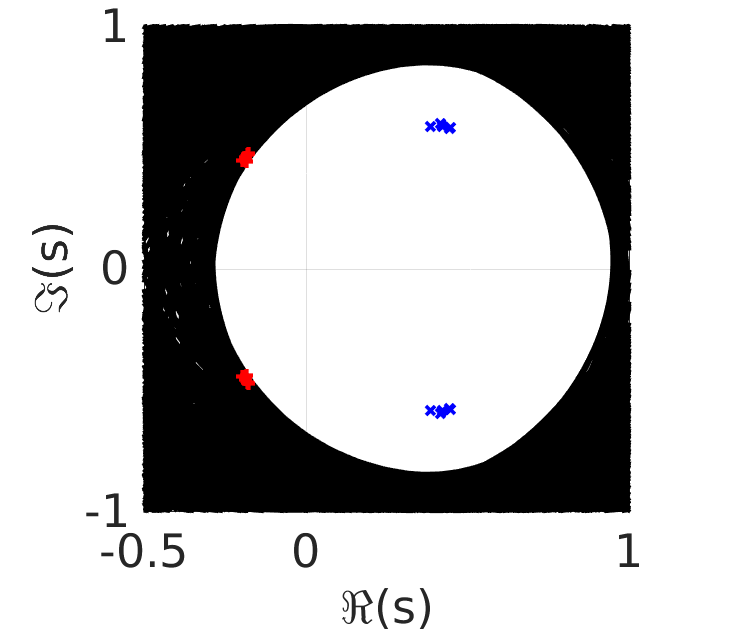}
				\caption{ SNR = 0dB.}
				\label{Fig:A1_B1_15dB} 
			\end{subfigure}
			~
			\begin{subfigure}[t]{0.22\textwidth}
				\centering
				\includegraphics[width =\textwidth]{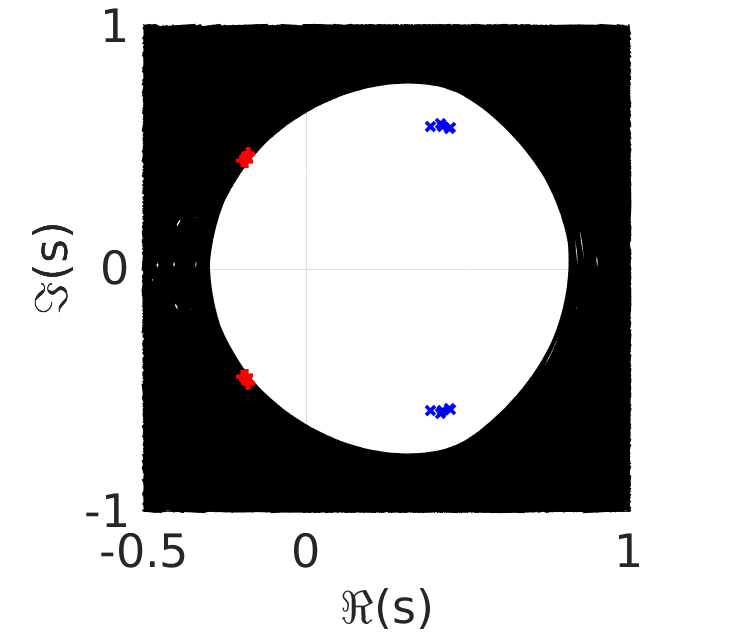}
				\caption{SNR = 10dB.}
				\label{Fig:A1_B1_25dB} 
			\end{subfigure}
			~				
			\begin{subfigure}[t]{0.22\textwidth}
				\centering
				\includegraphics[width =\textwidth]{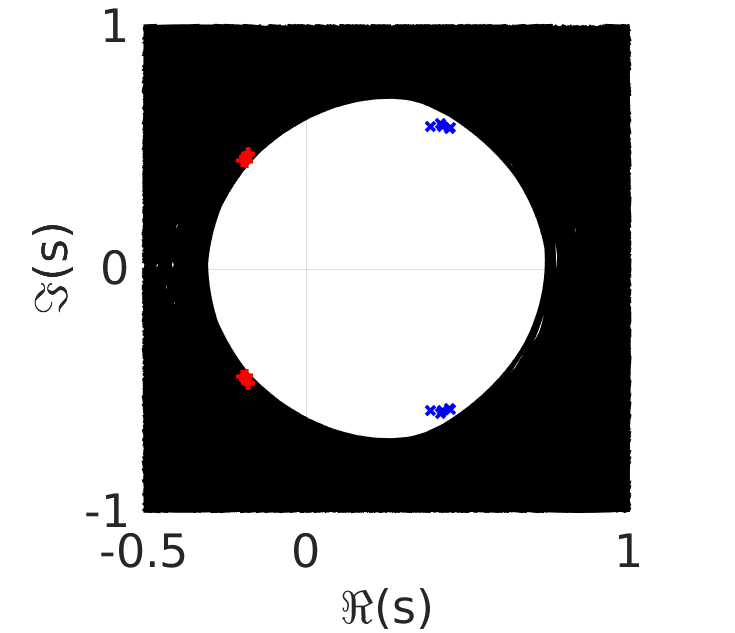}
				\caption{SNR = 20dB.}
				\label{Fig:A1_B1_30dB} 
			\end{subfigure}
			~
			\begin{subfigure}[t]{0.22\textwidth}
				\centering
				\includegraphics[width =\textwidth]{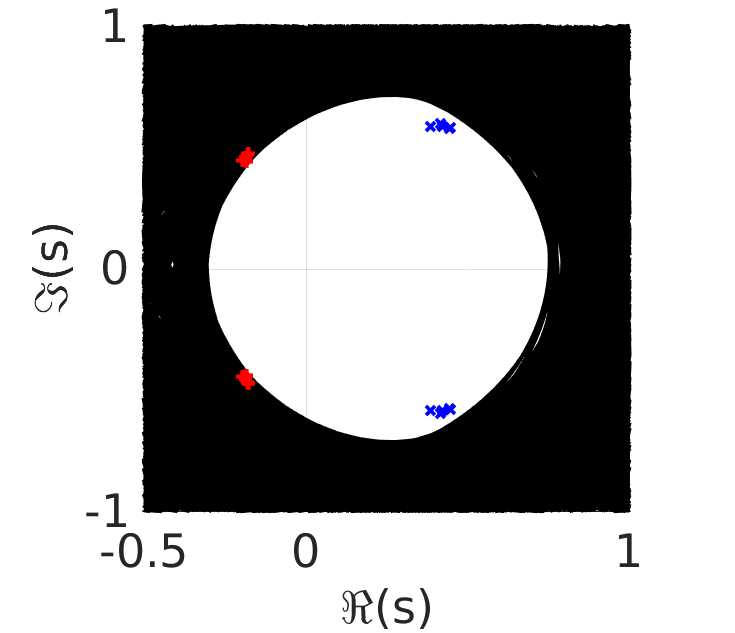}
				\caption{SNR = 30dB.}
				\label{Fig:A1_B1_35dB} 
			\end{subfigure}
			\caption{Numerical Range \Ww$(A_{\y},B_{\y})$ for different SNR.  Sets $Z_1 (\text{blue }\times)$, $Z_2(\text{red }+)$.}
			\label{Fig:A1_B1}
		\end{figure}

We compare the performances of CRNR and GLRT by computing the error rates. Fig. \ref{Fig:Comparation_MPM_RNR} shows the results. For low SNRs, from -5dB to 0dB, GLRT and CRNR have similar performances. However, when the SNR increases, there is approximately 2 dB difference between both methods.

		\begin{figure}[ht]
			\centering
			\includegraphics[width = 0.6 \textwidth]{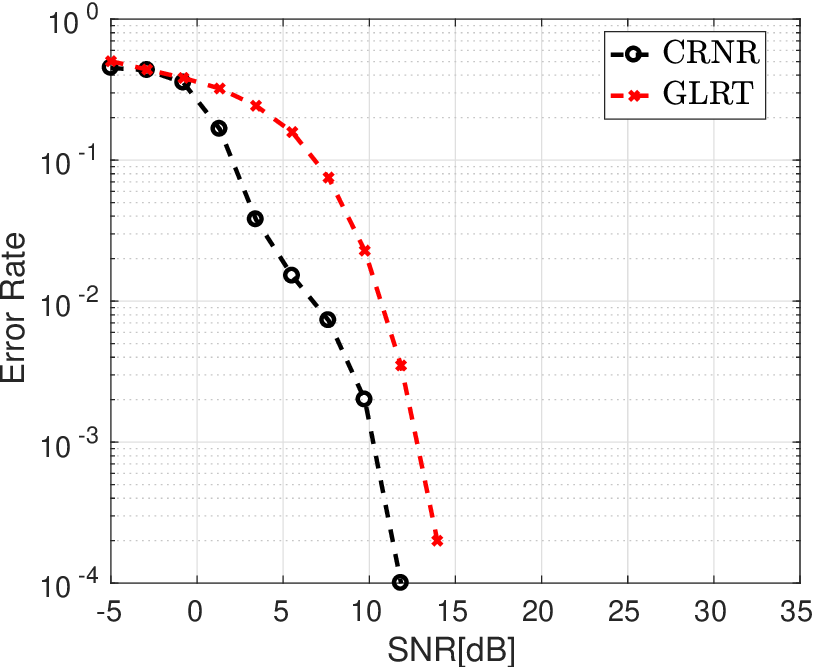}
			\caption{Error rates for observed signal $\y_t\in Z_1$ and classification set $Z_2$.}
			\label{Fig:Comparation_MPM_RNR}
		\end{figure}

\subsection{An unknown number of components}		


In this section, we compare the performance of the CRNR with the GLRT when $p\neq M$. The observed signal is as before built from $Z_1$. However, in this case, we estimate the order of the model using \cite{Gavish2014} and with this order we consider the number of complex frequencies for the classification test.
As before, we run a Montecarlo experiment with 10000 realizations.  Both methods are applied to each realization. The results are shown in Fig. \ref{Fig:Comparation_GLRT_RNR_1}. We observe that the CRNR outperforms the GLRT, which suffers when a wrong model order is considered. 		
		
		\begin{figure}[ht]
			\centering
			\includegraphics[width = 0.6\textwidth]{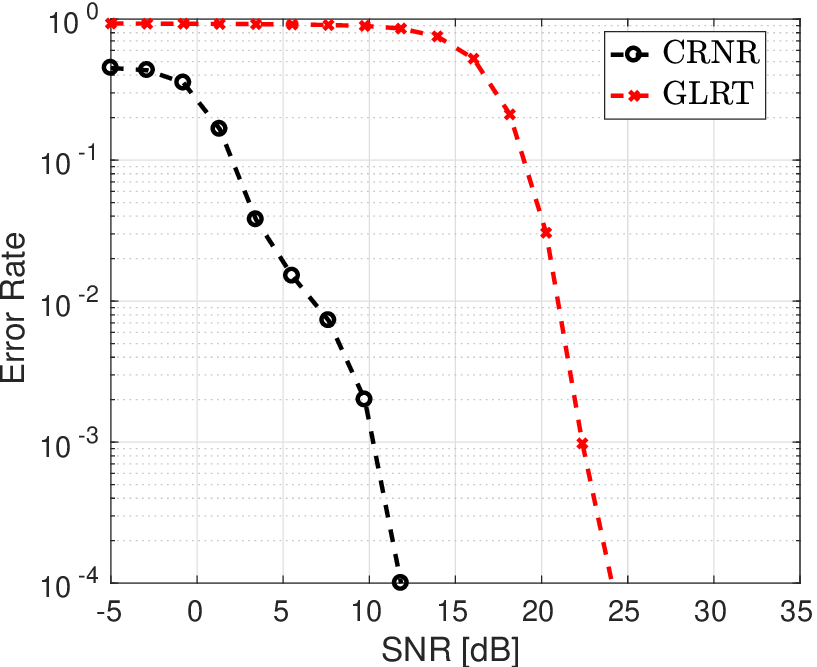}
			\caption{Error rates for observed signal $\y_t$ and candidate set $Z_2$.}
			\label{Fig:Comparation_GLRT_RNR_1}
		\end{figure}		
		
\subsection{Scattering problem}\label{Sect:SEM}
	
A scattering phenomenon occurs when an electromagnetic wave strikes a small object and the propagation continues over new directions, different from the incident one. Accurate mathematical models are very complex as they require the solution of Maxwell's equations for appropriate boundary conditions. The author in \cite{Baum1971} introduced an approximation known as the singularity expansion method (SEM) to describe the scattered field arriving at a given point in space as the combination of two components: a first one dominated by direct reflexions from the object, called the early-time component; a second one, which is due to induced surface currents on the object or cavity waves, called the late-time component. The SEM models the signal scattered from a target and observed from different angles as follows:	
		\begin{equation}
			\y_t = \sum_{i = 1}^{M} \mathbf{c}_iz_i^{(t-\tau_i)} + \hat{\y}_t, \quad t = 0,1,\ldots
			\label{Eq:SEM}
		\end{equation}
\noindent where $z_i$ and $\mathbf{c}_i$ are the complex natural resonances of the scattering object and the associated vectors of residues, $\tau_i$ is a constant delay, usually unknown, and $\hat{\y}_t$ is an entire function that represents the early-time component.  

Notice that the residues associated with each resonant modes are aspect dependent. This implies that some natural mode may not be significantly observed when its associated residue has very low energy or even zero along certain observation angle. Having multiple observation angles will be beneficial to overcome this difficulty.  

Internal resonances are associated with the composition of the particular object being illuminated. For example, Fig. \ref{Figure_Mie_signal} shows the backscattered signals for two different spheres of the same radius and different material. In the past, it has been proposed to use internal resonances to classify objects using scattering fields \cite{Bannis2014,Lee2012,Turhan-Say2003}.  All of these methods work in the resonance region, namely the late-time component of the signal. However, a clear separation between the early-time and the late-time components is not possible \cite{Pearson1982}. Usually, involved frequency-time techniques are implemented to achieve the signal separation in these two portions \cite{Chen2014}.

		\begin{figure}
			\centering
			\begin{subfigure}[t]{0.45\textwidth}
				\centering
				\includegraphics[width=\textwidth]{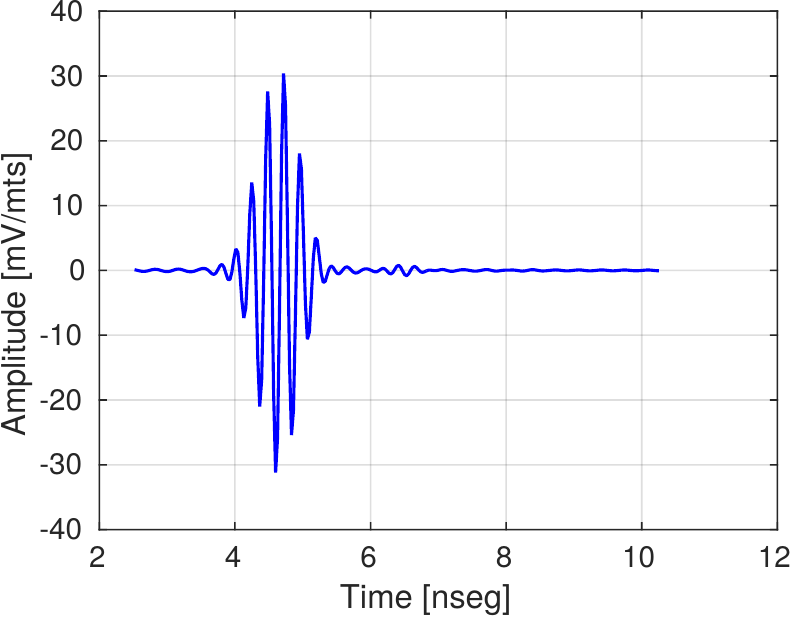}
				\caption{$\varepsilon_{1} = 2.12 - 0.053j$.}
				\label{Figure_Mie_signal_1}
			\end{subfigure}
			~
			\begin{subfigure}[t]{0.45\textwidth}
				\centering
				\includegraphics[width=\textwidth]{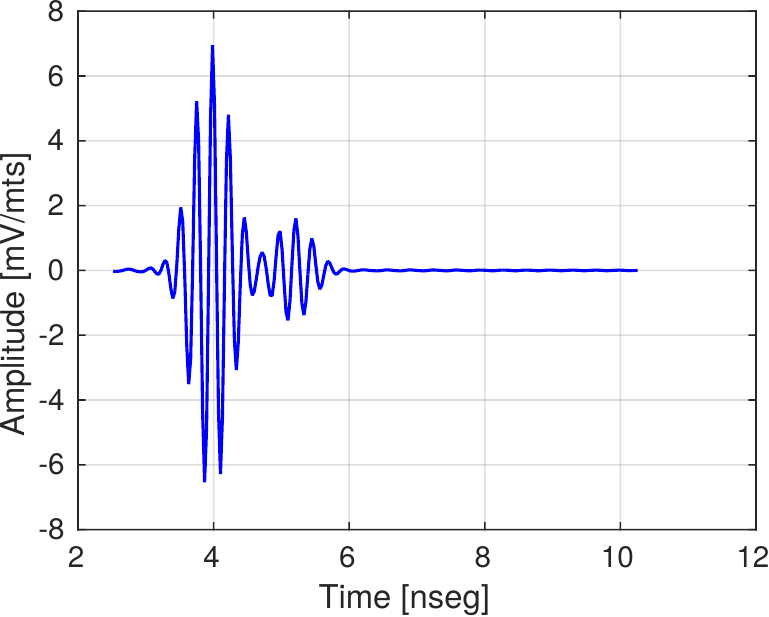}
				\caption{$\varepsilon_{2} = 7.0 - 5.25j$.}
				\label{Figure_Mie_signal_2}
			\end{subfigure}
			\caption{Scattered electric field obtained from spheres made of different materials.}	
			\label{Figure_Mie_signal}
		\end{figure}
		
To perform our study, we will consider spherical objects made of a homogeneous material and illuminated by a plane wave. In this case, 
the solution to Maxwell's equations is known as the Mie series \cite{Bohren1983}. We use this result to simulate the scattered field from 
a sphere with radius, $a = 0.07$ meters made of dielectric material with dielectric constant $\varepsilon$. We observe the response to a Gaussian pulse in the frequency range $1GHz$ to $5.9GHz$ for $K=3$ different observation angles, $\phi = [0^\circ,45^\circ,180^\circ]$. To set the classification problem, we consider two different materials, Class 1 with $\varepsilon_{1} = 2.12 - 0.053j$, and Class 2  with $\varepsilon_{2} = 7.0 - 5.25j$.

Fig. \ref{Fig:NatModes} shows some of the natural frequencies of Class 1 and Class 2. Both classes have a very large number of natural frequencies. However, it is possible to select disjoint subsets to characterize each class. The appropriate selection of the candidate modes is a delicate problem that determines the performance of the selection technique. In this paper, we choose a subset of $p=10$ modes to make the candidate class $Z_{\varepsilon_{2}}$ as shown in Fig. \ref{Fig:Natmodes_Sphere2}.
		
		\begin{figure}
			\centering
			\begin{subfigure}[t]{0.45\textwidth}
				\centering
				\includegraphics[width = \textwidth]{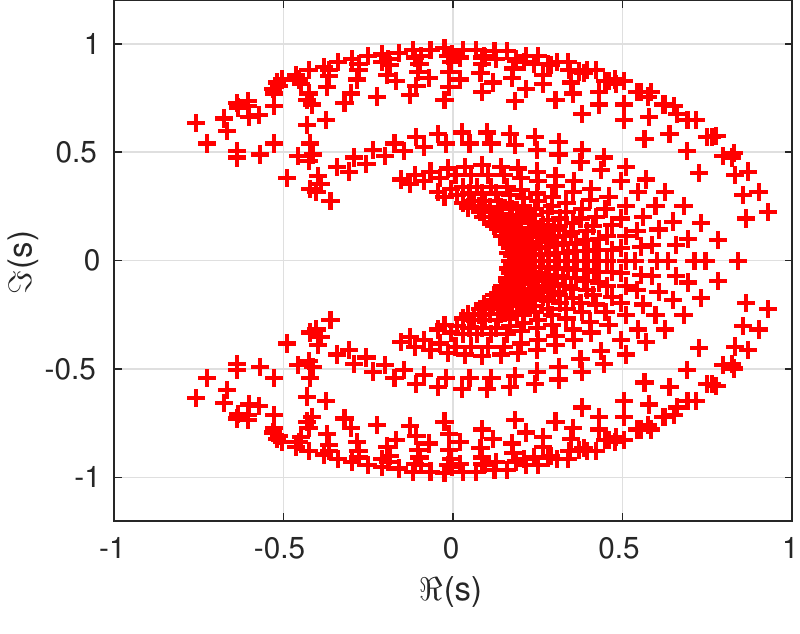}
				\caption{Natural modes for Sphere 1.}
				\label{Fig:Natmodes_Sphere1}
			\end{subfigure}
			~
			\begin{subfigure}[t]{0.45\textwidth}
				\centering
				\includegraphics[width = \textwidth]{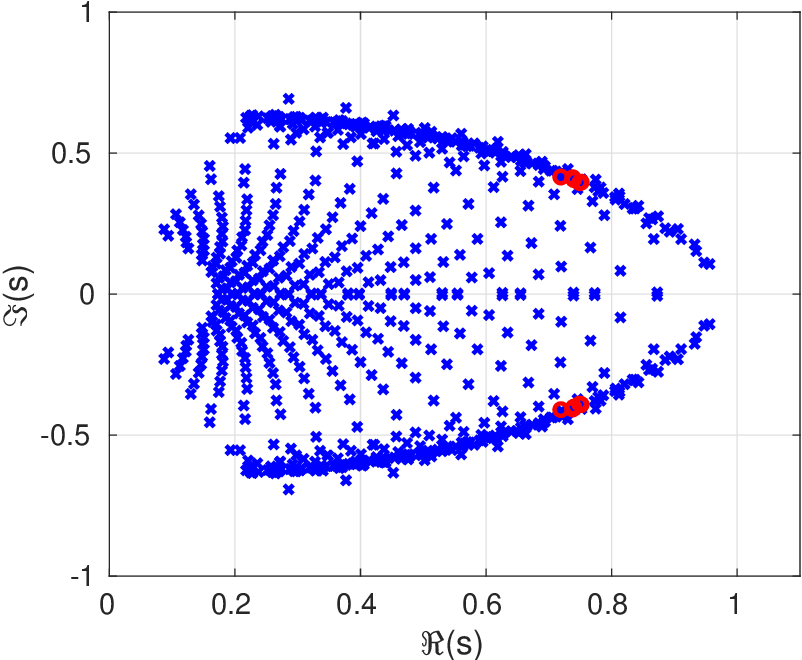}
								\caption{Natural modes ($\times$-blue) for Sphere 2, and $Z_{\varepsilon_{2}}$ ($\circ$-red).}
				\label{Fig:Natmodes_Sphere2}
			\end{subfigure}
			\caption{Natural modes for Sphere 1 and 2}
			\label{Fig:NatModes}
		\end{figure}

Using the observed signal from Sphere 1 and the classification set $Z_{\varepsilon_{2}}$, we follow Algorithm \ref{Algo}. 
We compare the performance of CRNR with the GLRT approach. In this case, we have used $p= 10$ known resonances to build the test. Fig. \ref{Fig:Error_rate_MIE} shows the classification error rates obtained for these methods.
The CRNR  outperforms the GLRT, which shows a poor performance when a reduced order model is considered. 

		\begin{figure}[ht]
			\centering
			\includegraphics[width = 0.6\textwidth]{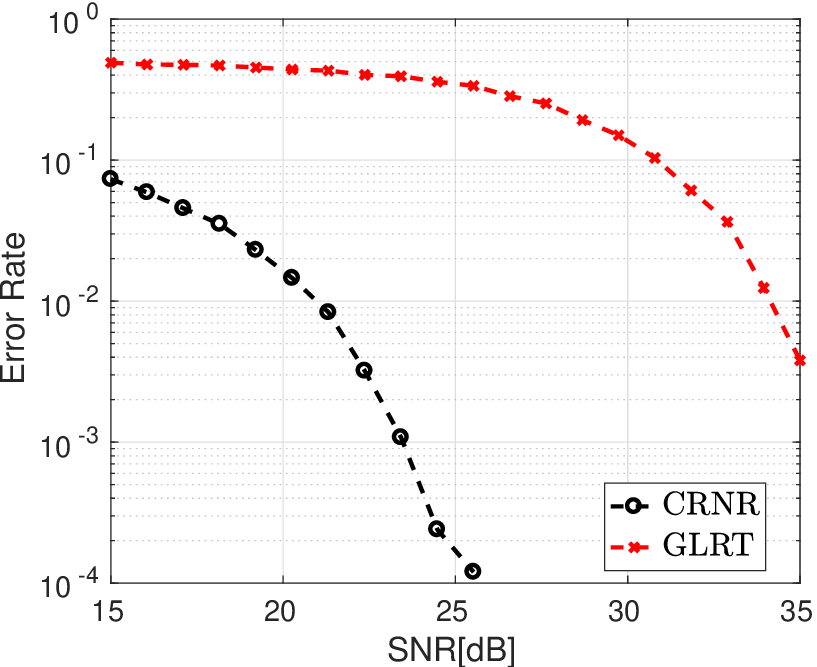}
			\caption{Error rate for observation pair $(A_{\varepsilon_{1}},B_{\varepsilon_{1}})$ and candidate set $Z_{\varepsilon_{2}}$}
			\label{Fig:Error_rate_MIE}
		\end{figure}

\section{Conclusion}

In this work, we have proposed a novel classification algorithm for signals composed of sums of exponentially damped sinusoids. Instead of estimating the frequencies and compared them to theoretical frequencies in a database, we have formulated an inverse problem. That is, given a set of frequencies, we evaluate if they correspond to the signal that we are observing. For that, we test the boundaries of the rectangular numerical range of the matrix pencil associated with the observed signal. With this proposition, we overcome some drawbacks observed in the matrix pencil method, just like the instability of the computation of the eigenvalues and the estimation of the model order.

We have applied this new scheme to the classification of dielectric material using the backscattered response of spherical objects when illuminated with ultra-wideband pulses. We have shown that the CRNR shows a good performance. As we mentioned, the resonant frequencies are independent of the aspect angle and polarization but their associated residues are not. So, theoretically we will have a unique set of infinite resonant frequencies associated with an object but when we measure the scattering signals we will observe only a subset of resonant frequencies. This makes the choice of candidate frequencies a critical step for the classification scheme. Nevertheless, the proposed scheme allows reducing by at least an order of magnitude the classification error compare to other existing techniques.


%

\appendix
\section{Reduced rank Hankel matrix approximation}\label{appendix1}
In this section, we summarize the procedure presented in \cite{Hua1991,Li1997}. Consider $H$ as in section \ref{Sect:MPM}, and define its singular value decomposition 
		\begin{equation}
			H = \sum_{i=1}^{n+1} \sigma_i\u_i\vv_i^H,
			\label{Eq:MPM11}
		\end{equation}
		\noindent where $\u_i\in\C^m$, $\vv_i\in\C^{n+1}$, and $\sigma_1\geq\sigma_2\geq\cdots\geq\sigma_{n+1}\geq 0$. 

Under the presence of noise, $H$ is full rank and has Hankel structure. The idea is to obtain a reduced-rank approximation that preserves the Hankel structure. This is important since it has been noticed that there is a one-to-one correspondence between data sequences consisting of superimposed damped sinusoids and rank-deficient Hankel matrices \cite{Li1997}.

Let us define two matrix operators. The first one, performs an $M$-order rank reduction, i.e., $\mathcal{L} : \C^{m\times (n+1)}\to \C^{m\times (n+1)}$,
		\begin{equation}
			\mathcal{L}\{H\} = \sum_{i=1}^{M}\sigma_i\u_i\vv_i^H. 
			\label{Eq:MPM5}
		\end{equation}
The second operator, obtains a Hankel approximation, i.e., $\mathcal{T}:\C^{m\times (n+1)}\to \C^{m\times (n+1)}$,
$			Y = \mathcal{T}\{X\},$
		\noindent where the $(k,l)-$th element of $Y$ ($k= 1,\ldots, m$, $l= 1,\ldots, n+1$) is obtained by averaging the elements of the anti-diagonal of $X$ as 
		\begin{equation}
			y_{k,l}  = \frac{1}{|\Lambda_{k+l}|}\sum_{(k',l')\in\Lambda_{k+l}}x_{k'l'}.
			\label{Eq:MPM7}
		\end{equation}

		\noindent Here, $\Lambda_{k+l}$ is the set of indices  corresponding to the $(k+l)$-th anti-diagonal of $X$, and $|\Lambda_{k+l}|$ is the cardinality of this set. 

		To obtain a rank-deficient Hankel matrix we use an iterative approach by applying alternately these two operators. First, we make a low-rank approximation and then we apply the Hankel approximation procedure. The procedure is repeated until the following stopping criterion is met

		\begin{equation}
			\|(\mathcal{T}\mathcal{L})^r\{X\} - \mathcal{L}(\mathcal{T}\mathcal{L})^{r-1}\{X\}\|_F < \epsilon,
			\label{Eq:MPM8}
		\end{equation}  

		\noindent where $r$ is the number of iterations performed and $\epsilon$ is a constant previously defined. 

\section*{Acknowledgment}
Raymundo Albert  is recipient of a doctoral scholarship from CONICET - Argentina. This work was part of the projects UBACyT 20020130100751BA, funded by the University of Buenos Aires, and  PIP-112 201101 00997, funded by CONICET . 





%
\section*{References}
\bibliography{bibliografia.bib} 

%

%
%
%




\end{document}